 \documentclass[letterpaper, 10 pt, conference]{ieeeconf}
 \IEEEoverridecommandlockouts

\usepackage[backend=biber, style=ieee]{biblatex}
\AtEveryBibitem{\ifthenelse{\ifentrytype{book}
    \OR\ifentrytype{article}
    \OR\ifentrytype{thesis}
    \OR\ifentrytype{inproceedings}
    \OR\ifentrytype{inbook}}
  {\clearfield{url}
    \clearfield{urldate}
    \clearfield{review}
    \clearfield{isbn}
    \clearfield{note}
    \clearfield{doi}
    \clearfield{series}}{}}

\usepackage[boxed]{algorithm2e}
\usepackage{bm}
\usepackage{amsmath,amssymb,latexsym}
\newtheorem{theorem}{Theorem}
\newtheorem{lemma}[theorem]{Lemma}
\newtheorem{assumption}[theorem]{Assumption}
\newtheorem{definition}[theorem]{Definition}

\newtheorem{remark}[theorem]{Remark}

\usepackage{graphicx}
\usepackage{booktabs}
\usepackage{longtable}
\allowdisplaybreaks 

\usepackage{tikz}
\usetikzlibrary{graphs}
\usepackage{tikz}
\usetikzlibrary{calc, backgrounds}
\usetikzlibrary{shapes.geometric, arrows}
\usepackage{circuitikz}
\usepgflibrary{arrows}
\usepackage{subcaption}

\ifnum\pdfshellescape=1
\fi
\usetikzlibrary{positioning}
\newcommand{\Hcal}{\mathcal{H}}
\newcommand{\I}{\mathbf{I}}
\newcommand{\0}{\mathbf{0}}
\newcommand{\Gcal}{\mathcal{G}}
\newcommand{\Ocal}{\mathcal{O}}
\newcommand{\Tcal}{\mathcal{T}}
\newcommand{\Ncal}{\mathcal{N}}

\newcommand{\RR}[1]{\mathbb{R}^{#1}}
\title{
	Data-driven Unknown-input Observers and State Estimation
}

\author{Mustafa Sahin Turan, Giancarlo Ferrari-Trecate
	\thanks{Authors are with the Institute of Mechanical Engineering (IGM),
		 \'Ecole Polytechnique F\'ed\'erale de Lausanne (EPFL),
		Switzerland. Email: {\tt\small  \{mustafa.turan, giancarlo.ferraritrecate\}@epfl.ch}}
	\thanks{This work has been supported by the Swiss National Science Foundation under the COFLEX project (grant number 200021\_169906) and the National Centre of Competence in Research (NCCR) in Dependable and Ubiquitous Automation.}
}

\begin{document}
\maketitle

\begin{abstract}
Unknown-input observers (UIOs) allow for estimation of the states of an LTI system without knowledge of all inputs.
In this paper, we provide a novel data-driven UIO based on behavioral system theory and the result known as Fundamental Lemma proposed by Jan Willems and coworkers. We give necessary and sufficient conditions on the data collected from the system for the existence of a UIO providing asymptotically converging state estimates, and propose a purely data-driven algorithm for their computation. Even though we focus on UIOs, our results also apply to the standard case of completely known inputs. As an example, we apply the proposed method to distributed state estimation in DC microgrids and illustrate its potential for cyber-attack detection.
\end{abstract}\vspace{-0.1cm}

\section{Introduction}\label{sec:Introduction}

The problem of estimating the states of an LTI system when some inputs cannot be measured has been studied within the control community for almost half a century~\cite{meditch1973observers}, and has been motivated by applications in control, robust estimation, and fault diagnosis.
Among approaches available in the literature, some use \textit{a priori} information about the unknown inputs $d$, whereas some others assume no such prior and develop \textit{unknown-input decoupling} observers, i.e., state estimators whose estimation error is independent of $d$ and asymptotically converges to zero~\cite{darouach1994full}. In this paper, we focus on the latter class. Such observers, called unknown-input observers (UIOs) from now on, have been developed for continuous-time~\cite{darouach1994full}
and discrete-time systems~\cite{valcher1999state}. UIOs are often used for fault detection~\cite{chen1996design, gao2015unknown} and, more recently, for cyber-attack detection~\cite{gallo2018distributed, gallo2020distributed}. They are an attractive tool in \textit{remote} and \textit{distributed} settings, where state estimators are not collocated with the system, and therefore, do not have access to all its inputs.

The work~\cite{valcher1999state} provides necessary and sufficient UIO existence conditions based on system matrices, which represent suitable observability and decoupling properties of the system. It also gives a model-based UIO design procedure under these conditions. However, the literature lacks end-to-end methodologies using data instead of a system model. In particular, no existing work provides a data-driven formulation of UIO existence conditions and design. An approach to achieve this goal is to follow a two-step procedure by first identifying the system from the collected data and then designing a UIO for the reconstructed model. 

Among the techniques for identifying systems with fully- or partially-unknown inputs, subspace identification can be used when $d$ is a zero-mean stationary white noise~\cite{van2012subspace}. Similarly, errors-in-variables (EIV) 
methods can be applied when the unknown input can be modeled by additive stationary noise perturbing known input variables~\cite{soderstrom2007errors}.
In order to remove the above assumptions on the unknown inputs, recently, the \textit{indirect framework} has been proposed~\cite{linder2017identification}. The goal is achieved by introducing system-level assumptions ensuring that some inputs can be directly measured, and certain parts of the system dynamics are known. The element-level system identification method proposed in~\cite{wang1994element} does away with assumptions on the system or the unknown inputs, but restricts the focus on mechanical systems. Similarly,~\cite{yu2016blind} proposes a \textit{blind subspace identification} scheme under the assumption of persistently exciting unknown inputs.

We highlight that, except~\cite{linder2017identification}, none of the above methods guarantees the exact identification of the system with finite data, even without noise in the measured variables. Moreover, conditions for the existence of a UIO are rank-based~\cite{valcher1999state}, and therefore extremely sensitive to errors in the identified system matrices. Identification errors may also result in poor estimation performance of UIOs as input-decoupling conditions are inherently sensitive to uncertainties in the system matrices. In addition,~\cite{linder2017identification} does not identify the input channels corresponding to the unknown inputs, hampering the application of existing model-based UIO design methods. 

An alternative to the two-step approach
is to check the existence of a UIO and design the observer directly from data, without building a model of the system. 

In this paper, we propose a method with these features by exploiting the \textit{Fundamental Lemma}~\cite{willems2005note}, a key result in behavioral system theory showing that all trajectories of a linear system can be spanned by a finite number of input-output samples. The Fundamental Lemma has been used for developing 
data-driven simulation and output prediction~\cite{markovsky2008data}, stability analysis and state-feedback control design~\cite{de2019formulas}, predictive control~\cite{coulson2019data, berberich2020data}, and robust optimal control~\cite{xu2021non, xu2021data}. In particular, we exploit the results in~\cite{de2019formulas, markovsky2008data} to give necessary and sufficient conditions for the existence of a UIO and develop a design procedure. In this paper, we consider the case of noiseless data 
but, unlike~\cite{van2012subspace, soderstrom2007errors, linder2017identification}, we do not assume any knowledge of the system dynamics or the process generating $d$. 
Moreover, our results can be directly extended to standard state estimation with no unknown inputs. As an application example, we apply the proposed UIO to DC microgrids (DCmGs), and show how it can be used for distributed cyber-attack detection. 

This paper is organized as follows. Section~\ref{sec:ProblemFormulation} formally presents the problem, while the UIO design is discussed in Section~\ref{sec:DataDrivenUIO}. The application example is given in Section~\ref{sec:MicrogridExample}, before concluding the paper in Section~\ref{sec:Conclusions}.\vspace{-0.2cm}

\subsection*{Notation}\label{subsec:Notation}\vspace{-0.1cm}
$\I_n \in \RR{n}$ and $\0_{n\times m} \in \RR{n\times m}$ represent identity and zero matrices, respectively. 
For a matrix $A \in \mathbb{R}^{n\times m}$, $\ker(A)$ and $\mathrm{range}(A)$ denote its null and column spaces, respectively. 
$A^\dagger$ is used for the pseudoinverse of $A$.
For a sequence ${\{v_k\}}_{k=i}^j$ of vectors, $\mathrm{col}(\{v_k\}_{k=i}^j)$ stands for the column concatenation of the vectors $v_k$ and the resulting vector is denoted by $v_{[i:j]}$.
When the start and end indices $i,~j$ are clear from the context, we use $v$ instead.
The \textit{Hankel matrix of depth} $L$ associated to $v_{[i:j]}$, $j\geq i+L-1$, is defined as \vspace{-0.1cm}
\begin{align*}\vspace{-0.1cm}
\mathcal{H}_{L}(v)\triangleq\left[\begin{array}{cccc}v_{i} & v_{i+1} & \cdots & v_{j-L+1} \\ v_{i+1} & v_{i+2} & \cdots & v_{j-L+2} \\ \vdots & \vdots & \ddots & \vdots \\ v_{i+L-1} & v_{i+L} & \cdots & v_{j}\end{array}\right].
\end{align*}
The sequence ${\{v_k\}}_{k=i}^j$ is called \textit{persistently exciting of order $L$} if $\mathcal{H}_{L}(v)$ has full row rank.

\section{Problem Formulation}\label{sec:ProblemFormulation}
Consider a system $\Gcal$ with the state-space representation\vspace{-0.1cm}
\begin{equation} \label{eq:LTI_dynamics}\vspace{-0.1cm}
	\begin{split}
	  x_{t+1} &= Ax_t + B u_t + Ed_t, \\
	  y_t &= Cx_t,
	\end{split}
\end{equation}
where $x_t\in\RR{n}$ are the states, $u_t\in\RR{m}$ are the (known) inputs, $y_t\in\RR{p}$ are the outputs, and $d_t\in\RR{m_d}$ are the \textit{unknown} inputs (e.g. disturbances) of the system, and hence unmeasured. In this paper, we assume that the system is in minimal form, i.e., $(A,[B~E])$ is controllable and $(A,C)$ is observable. 

\begin{definition}[UIO~\cite{valcher1999state}]\label{def:UIO}
	An LTI system of the form \vspace{-0.10cm}
	\begin{equation} \label{eq:UIO_dynamics_general}\vspace{-0.10cm}
	\begin{split}
	z_{t+1} &= A_\mathrm{UIO}z_t + B_\mathrm{UIO} v_t, \\
	\hat{x}_t &= z_t + D_{\mathrm{UIO}} v_t,
	\end{split}
	\end{equation}
	with inputs $v\triangleq [u^\top~y^\top]^\top$ and outputs $\hat{x}$ is a UIO for the system in~\eqref{eq:LTI_dynamics} if $\hat{x}_t-x_t\rightarrow 0$ as $t\rightarrow \infty$ for any initial states $x_0$ and $z_0$, input $u$, and unknown input $d$.
\end{definition}

\begin{remark}\label{rem:standard_state_estimation}
	When $m_d=0$, the formulation in~\eqref{eq:LTI_dynamics} and~\eqref{eq:UIO_dynamics_general}, as well as the following analysis, capture standard state-estimation problems where all inputs are known. 
\end{remark}

If the matrices $A$, $C$, and $E$ of the system~\eqref{eq:LTI_dynamics} satisfy certain unknown-input observability conditions, a UIO exists~\cite{valcher1999state}.

\begin{remark}\label{rem:UIO_traj}
	As shown in~\cite{valcher1999state}, if a UIO can be designed, the state-estimation error $e_t\triangleq x_t-\hat{x}_t$ follows the autonomous dynamics $e_{t+1}=A_\mathrm{UIO}e_t$. By setting the initial condition of the UIO as $z_0 = x_0-D_{\mathrm{UIO}} y_0$, one gets $\hat{x}_0=x_0$ and, consequently, $\hat{x}_t = x_t~~\forall t$. Therefore, for any input-output-state trajectory $(u, y, x)$ of $\Gcal$, $([u^\top~y^\top]^\top, x)$ is an input-output trajectory of the UIO~\eqref{eq:UIO_dynamics_general}.
\end{remark}

In the sequel, we assume that $x_0$ is not available and, thus, $z_0$ cannot be chosen as above. Regardless, the observation in Remark~\ref{rem:UIO_traj} is key in our approach as it enables us to collect data from the UIO without constructing it.

In order to provide a data-driven UIO formulation, 
we assume that an \textit{offline} experiment has been conducted with the system~$\Gcal$ before the start of any estimation task, and the corresponding input-output-state trajectories $\bar{u}\triangleq\mathrm{col}(\{\bar{u}_i\}_{i=0}^{T-1})$, $\bar{y}\triangleq\mathrm{col}(\{\bar{y}_i\}_{i=0}^{T-1})$, $\bar{x}\triangleq\mathrm{col}(\{\bar{x}_i\}_{i=0}^{T-1})$ have been collected. 
These data, named \textit{historical}, define the following matrices \vspace{-0.15cm}
\begin{equation}\label{eq:Hankel_IOS}\vspace{-0.15cm}
	\begin{split}
	U\triangleq \Hcal_L(\bar{u}),\enskip
	Y\triangleq \Hcal_L(\bar{y}),\enskip
	X\triangleq \Hcal_L(\bar{x}),
	\end{split}
\end{equation}
for some $L\leq T$. Similarly, define the Hankel matrix corresponding to $\bar{v}\triangleq\mathrm{col}(\{\bar{v}_i\}_{i=0}^{T-1}) = \mathrm{col}(\{[\bar{u}_i^\top~\bar{y}_i^\top]^\top\}_{i=0}^{T-1})$ as $V\triangleq\Hcal_L(\bar{v})$.
Although $d$ is not measured, we introduce the notation $\bar{d} \triangleq \mathrm{col} (\{\bar{d}_i\}_{i=0}^{T-1})$ for historical unknown input data. The corresponding Hankel matrix is $D\triangleq \Hcal_L(\bar{d})$.

When a UIO~\eqref{eq:UIO_dynamics_general} exists, results in~\cite{markovsky2008data} can be applied to predict its outputs, which is equivalent to computing state estimations. This methodology requires, at each time step $t$, to specify \textit{recent} data $v_{t,\mathrm{ini}}\triangleq\mathrm{col}(\{v_i\}_{i=t-T_\mathrm{ini}}^{t-1})$, $\hat{x}_{t,\mathrm{ini}}\triangleq\mathrm{col}(\{\hat{x}_i\}_{i=t-T_\mathrm{ini}}^{t-1})$ consisting of $T_\mathrm{ini}$ samples. This data uniquely determines the the state $z_{t-1}$ of the UIO 
if $T_\mathrm{ini}\geq l_\mathrm{UIO}$, where $l_\mathrm{UIO}$ is the observability index of the UIO. 
Algorithm~1 in~\cite{markovsky2008data} computes output predictions for a future horizon of $T_f$ samples based on the recent data and \textit{future} inputs $v_{t,f}\triangleq\mathrm{col}(\{v_i\}_{i=t}^{t+T_f-1})$. For this purpose, Hankel matrices are separated into \textit{past} and \textit{future} blocks denoted by subscripts $p$ and $f$, respectively:\vspace{-0.1cm}
\begin{equation}\label{eq:past_future_Hankel}\vspace{-0.1cm}
	U = \begin{bmatrix}
	U_p \\
	U_f
	\end{bmatrix},\enskip Y = \begin{bmatrix}
	Y_p \\
	Y_f
	\end{bmatrix}, \enskip X = \begin{bmatrix}
	X_p \\
	X_f
	\end{bmatrix}, \enskip V = \begin{bmatrix}
	V_p \\
	V_f
	\end{bmatrix},
\end{equation}
where the upper block matrices consist of $T_\mathrm{ini}$ block rows, and the lower block matrices consist of $T_f$ block rows. 
In this paper, we iteratively apply the abovementioned algorithm with one-step-ahead predictions (see Section~\ref{sec:DataDrivenUIO}); therefore, we take $T_f=1$. We also take $T_\mathrm{ini}=1$, since the output matrix of the UIO~\eqref{eq:UIO_dynamics_general} is identity, which implies $l_\mathrm{UIO}=1$.

In what follows, it is assumed that inputs and outputs of the system~$\Gcal$ are accessible. 
The states are considered to be measured in the offline experiment to collect the historical data, but not accessible in real-time operation. 

\begin{remark}\label{rem:available_data}
Our assumption on the availability of the states is often fulfilled in a remote estimation scenario, where the observer is not collocated with the system. As such, it might be impossible, unsafe, or unfeasible for the system to transmit the state measurements to the observer in real time over a communication network. Instead, the \textit{historical} state data can be collected offline and transferred once and for all to the observer by using a different physical medium. 
Moreover, historical states can be measured once in dedicated lab experiments using sensors that can be costly to install in real-time applications. As cost reduction is a key driver in industry~\cite{mastellone2021impact}, it might be desirable to estimate states in online operations instead of adding sensors, especially if several copies of the same system are created. Finally note that infinitely many state-space realizations of~$\Gcal$ exist~\cite{van2012subspace}. In order to estimate the states of~$\Gcal$ uniquely in the absence of model knowledge, it is required to fix their basis, which is achieved by historical state measurements.
\end{remark}

\begin{definition}\label{def:compatible_trajectory}
	A trajectory $(\{v_i\}_{i=0}^{N-1},\{x_i\}_{i=0}^{N-1})$ is compatible with the historical data if
	\begin{equation}\label{eq:compatible_trajectory}
		\begin{bmatrix}
			v_i \\
			x_i \\
			v_{i+1} \\
			x_{i+1}
		\end{bmatrix} \in \mathrm{range}\left(\begin{bmatrix}
			V_p \\
			X_p \\
			V_f \\
			X_f
		\end{bmatrix}\right), \enskip \forall i\in\{0, 1, \dots, N-2\}.
	\end{equation}
	Moreover, the set of all trajectories compatible with given historical data $(\bar{v}, \bar{x})$ is defined as\vspace{-0.2cm}
	\begin{equation}\label{eq:set_compatible_trajectories}\vspace{-0.2cm}
	\mathbb{T}_c(\bar{v}, \bar{x}) \triangleq \{(\{v_i\}_{i=0}^{N-1},\{x_i\}_{i=0}^{N-1})|~\eqref{eq:compatible_trajectory}~\text{holds}\}.
	\end{equation}
\end{definition}\vspace{0.3cm}

We further introduce the set of all trajectories $(\{v_i\}_{i=0}^{N-1},\{x_i\}_{i=0}^{N-1})$ that can be generated by $\Gcal$:
\begin{equation}\label{eq:set_system_trajectories}
\begin{split}
\mathbb{T}_\Gcal \triangleq \{(\{[u_i^\top~&y_i^\top]^\top\}_{i=0}^{N-1},\{x_i\}_{i=0}^{N-1})| ~\exists \{d_i\}_{i=0}^{N-1} \\
&~\text{verifying}~\eqref{eq:LTI_dynamics}~\forall i\in\{0,1,\dots,N-2\}\}.
\end{split}
\end{equation}

Definition~\ref{def:compatible_trajectory} and equation~\eqref{eq:set_system_trajectories} are used for checking whether the historical data are representative of all input-output trajectories of $\Gcal$. Note that this is achieved when all trajectories of $\Gcal$ are compatible with the historical data, i.e., $\mathbb{T}_\Gcal = \mathbb{T}_c(\bar{v}, \bar{x})$. Indeed, if historical trajectories are very short or poorly chosen, the range of $[V_p^\top~X_p^\top~V_f^\top~X_f^\top]^\top$ might be very small and incompatible trajectories of $\Gcal$ might exist. 

In this paper, we assume all data to be noiseless in order to provide the theory for data-driven UIO\footnote{Noiseless historical data, which corresponds to perfect model knowledge, and noiseless recent (online) data are standard assumptions in the setting in which Luenberger observer and UIOs were originally developed. 
}. 
As discussed in Remark~\ref{rem:available_data}, in certain applications, historical data can be generated in dedicated experiments. In such cases, historical data can be assumed noiseless when sophisticated and accurate sensors are used.
The presence of measurement noise in recent data is discussed later in Remark~\ref{rem:noisy_recent_data}.\vspace{-0.1cm}

\section{Data-Driven UIO}\label{sec:DataDrivenUIO}

In this section, we present the proposed data-driven UIO formulation. Our method is enabled by the observation in Remark~\ref{rem:UIO_traj} that the input-output-state trajectories of $\Gcal$ also represent input-output trajectories of the UIO. Therefore, if a UIO exists, historical data collected from $\Gcal$ can be used to provide a data-driven representation of the trajectories of the UIO, when $\mathbb{T}_\Gcal = \mathbb{T}_c(\bar{v},\bar{x})$~\cite{de2019formulas}.

The presentation of our results is structured in three steps. In Lemma~\ref{lem:historical_data_span}, we give a sufficient condition for having $\mathbb{T}_\Gcal = \mathbb{T}_c(\bar{v},\bar{x})$. In Lemma~\ref{lem:ExistsSchurSystem}, we present necessary and sufficient conditions for the existence of a system of the form~\eqref{eq:UIO_dynamics_general} that generates all trajectories in $\mathbb{T}_c(\bar{v},\bar{x})$. Finally, Theorem~\ref{thm:UIO_and_estimations} characterizes the existence of a UIO and provides a data-driven UIO estimation scheme. The following assumption is required in the sequel.
\begin{assumption}\label{ass:pers_excit}
	The historical data $\{[\bar{u}_i^\top~~\bar{d}_i^\top]^\top\}_{i=0}^{T-1}$ are persistently exciting of order $n+2$.
\end{assumption}

\begin{lemma}\label{lem:historical_data_span}
	If Assumption~\ref{ass:pers_excit} holds, $\mathbb{T}_c(\bar{v}, \bar{x})=\mathbb{T}_\Gcal$.
\end{lemma}

\begin{proof}
	Since $v_t = [u_t^\top~y^\top_t]^\top$, 
	there exists a row permutation matrix $P_R$ such that 
	$$P_R\begin{bmatrix}
	v_t \\
	x_t \\
	v_{t+1} \\
	x_{t+1}
	\end{bmatrix} = \begin{bmatrix}
	u_{[t:t+1]} \\
	y_{[t:t+1]} \\
	x_{[t:t+1]}
	\end{bmatrix},$$ for any vector $[v_t^\top~x_t^\top~v_{t+1}^\top~x_{t+1}^\top]^\top$ corresponding to a trajectory of $\Gcal$. From~\eqref{eq:LTI_dynamics}, the variables on the right-hand side of the above equation verify
	\begin{equation}\label{eq:uyx_udx0}
	\begin{split}
	\begin{bmatrix}
	u_{[t:t+1]} \\
	y_{[t:t+1]} \\
	x_{[t:t+1]}
	\end{bmatrix} = \underbrace{\begin{bmatrix}
		\I & \0 & \0 \\
		\Tcal_{uy,2} & \Tcal_{dy,2} & \Ocal_{y,2} \\
		\Tcal_{ux,2} & \Tcal_{dx,2} & \Ocal_{x,2}
		\end{bmatrix}}_{\triangleq \Theta}\begin{bmatrix}
	u_{[t:t+1]} \\
	d_{[t:t+1]} \\
	x_{t}
	\end{bmatrix},
	\end{split}
	\end{equation}
	where \vspace{-0.2cm}
	\begin{equation*}
	\vspace{-0.2cm}
	\begin{split}
	\Tcal_{uy,2} = \begin{bmatrix}
	\0 & \0 \\
	CB & \0
	\end{bmatrix}, \enskip \Tcal_{dy,2} &= \begin{bmatrix}
	\0 & \0 \\
	CE & \0
	\end{bmatrix}, \enskip \Ocal_{y,2} = \begin{bmatrix}
	C \\
	CA 
	\end{bmatrix}, \\
	\Tcal_{ux,2} = \begin{bmatrix}
	\0 & \0 \\
	B & \0
	\end{bmatrix}, \enskip
	\Tcal_{dx,2} &= \begin{bmatrix}
	\0 & \0 \\
	E & \0
	\end{bmatrix}, \enskip 
	\Ocal_{x,2} = \begin{bmatrix}
	\I \\
	A 
	\end{bmatrix}.
	\end{split}
	\end{equation*}
	Therefore, for any trajectory $(\{v_i\}_{i=0}^{N-1},\{x_i\}_{i=0}^{N-1},\{d_i\}_{i=0}^{N-1})$ of the system $\Gcal$, it holds that, for all $t\in\{0,\dots,N-2\}$, \vspace{-0.2cm}
	\begin{equation}\label{eq:traj_PR_Theta}\vspace{-0.2cm}
	\begin{split}
	\begin{bmatrix}
	v_t \\
	x_t \\
	v_{t+1} \\
	x_{t+1}
	\end{bmatrix} = P_R^{-1}\Theta\begin{bmatrix}
	u_{[t:t+1]} \\
	d_{[t:t+1]} \\
	x_{t}
	\end{bmatrix}.
	\end{split}
	\end{equation}
	Moreover, given a sequence of inputs $(\{u_i\}_{i=0}^{N-1}, \{d_i\}_{i=0}^{N-1})$ and an initial state $x_0$, any sequence $(\{v_i\}_{i=0}^{N-1},\{x_i\}_{i=0}^{N-1})$ obtained by iteratively solving for the left-hand side of~\eqref{eq:traj_PR_Theta} for $t\in\{0,\dots,N-2\}$ is a trajectory of $\Gcal$. 
	For any set of historical data $(\bar{v}, \bar{x}, \bar{d})$ generated by $\Gcal$, it holds that\vspace{-0.1cm}
	\begin{equation}\label{eq:comp_PR_Theta}\vspace{-0.1cm}
	\begin{split}
	\begin{bmatrix}
	V_p \\
	X_p \\
	V_f \\
	X_f
	\end{bmatrix} = P_R^{-1}\Theta\begin{bmatrix}
	U \\
	D \\
	X_p
	\end{bmatrix}
	\end{split}
	\end{equation}
	because every column of the left-hand side of the above equation is a trajectory of $\Gcal$. Therefore, $\mathbb{T}_c(\bar{v}, \bar{x})\subseteq \mathbb{T}_\Gcal$. We next show that $\mathbb{T}_\Gcal\subseteq\mathbb{T}_c(\bar{v}, \bar{x})$. For this, it is sufficient to verify that for any trajectory $(\{v_i\}_{i=0}^{N-1},\{x_i\}_{i=0}^{N-1})$ of $\Gcal$, every vector $[v_t^\top~x_t^\top~v_{t+1}^\top~x_{t+1}^\top]^\top$ is in the range of $[V_p^\top~X_p^\top~V_{f}^\top~X_{f}^\top]^\top$. Under Assumption~\ref{ass:pers_excit}, Theorem~1 in~\cite{van2020willems} can directly be applied to show that $[U^\top~D^\top~X_p^\top]^\top$ has full row rank. As a direct consequence, given a vector $[u_{[t:t+1]}^\top~d_{[t:t+1]}^\top~x_{t}^\top]^\top$, there exists a vector $g_{t+1}$ such that \vspace{-0.1cm}
	$$\begin{bmatrix}
		U \\
		D \\
		X_p
		\end{bmatrix}g_{t+1} = \begin{bmatrix}
		u_{[t:t+1]} \\
		d_{[t:t+1]} \\
		x_{t}
		\end{bmatrix}.$$\vspace{-0.1cm} Then, multiplying~\eqref{eq:comp_PR_Theta} from the right by $g_{t+1}$ yields $$\begin{bmatrix}
		V_p \\
		X_p \\
		V_f \\
		X_f
		\end{bmatrix} g_{t+1} = \begin{bmatrix}
		v_t \\
		x_t \\
		v_{t+1} \\
		x_{t+1}
		\end{bmatrix},$$ where the vector $[v_t^\top~x_t^\top~v_{t+1}^\top~x_{t+1}^\top]^\top$ satisfies~\eqref{eq:traj_PR_Theta}. Since any trajectory of $\Gcal$ consists of $v_t$, $x_t$, $v_{t+1}$, $x_{t+1}$ satisfying~\eqref{eq:traj_PR_Theta}, one gets $\mathbb{T}_\Gcal\subseteq\mathbb{T}_c(\bar{v}, \bar{x})$.
\end{proof}

\begin{remark}\label{rem:d_pers_excit}
	Lemma~\ref{lem:historical_data_span} requires persistency of excitation of the unknown inputs $\bar{d}$, which is not verifiable using the available data. This assumption can be satisfied when the unknown inputs cannot be measured or modified, but change randomly. It is also satisfied if $\bar{d} = \bar{d}_0 + \delta\bar{d}$, where $\bar{d}_0$ is a (not necessarily exciting) deterministic component and $\delta\bar{d}$ is a small random component. For example, in DCmGs, unknown inputs include the current loads connected to generation units (see Section~\ref{sec:MicrogridExample}). Loads are dictated by current consumption which can be assumed to have a random component\footnote{Loads might include aggregated domestic consumption based on complex daily activity patterns of many consumers, which can be assumed stochastic. Load currents are also affected by noise terms that are induced by switches in power-electronics converters used for connecting loads.}. Persistency of excitation can also be satisfied when $\delta\bar{d}$ belongs to certain classes of deterministic signals such as pseudo-random binary sequences (PRBSs), and sums of sinusoids~\cite[Chapter~5]{soderstrom1989system}.
\end{remark}

In the following, we make use of the vector $g_{t+1}$ solving \vspace{-0.1cm}
\begin{equation}\label{eq:g_sol}\vspace{-0.1cm}
\begin{bmatrix}
V_p \\
X_p \\
V_f
\end{bmatrix}g_{t+1} = \begin{bmatrix}
v_t \\
x_t \\
v_{t+1}
\end{bmatrix}
\end{equation} 
for given $V_p$, $X_p$, $V_f$ and a compatible recent trajectory $v_{[t:t+1]},$ $x_t$. All solutions to~\eqref{eq:g_sol} can be written as\vspace{-0.1cm}
\begin{equation}\label{eq:g_sol_aff}\vspace{-0.1cm}
g_{t+1}=\Xi[v_t^\top~x_t^\top~v_{t+1}^\top]^\top+\nu,
\end{equation}
for a vector $\nu \in\ker([V_p^\top~X_p^\top~V_f^\top]^\top)$ and a properly defined matrix $\Xi$. There are infinitely many such matrices and a particular choice is $([V_p^\top~X_p^\top~V_f^\top]^\top)^\dagger$.
We partition this matrix as $\Xi=[\Xi_{V_p} ~ \Xi_{X_p} ~ \Xi_{V_f}]$, where $\Xi_{V_p}$, $\Xi_{X_p}$, and $\Xi_{V_f}$ have $m+p$, $n$, and $m+p$ columns, respectively.

\begin{lemma}\label{lem:ExistsSchurSystem}
	There exists an LTI system of the form~\eqref{eq:UIO_dynamics_general} 
	that can generate every compatible input-output trajectory $(\{v_i\}_{i=0}^{N-1},\{x_i\}_{i=0}^{N-1})$ if and only if \vspace{-0.1cm}
	\begin{equation}\label{eq:Hankel_existence_condition}\vspace{0.1cm}
	\ker\left(\begin{bmatrix}
	V_p \\
	X_p \\
	V_f
	\end{bmatrix}\right) \subseteq \ker(X_f).
	\end{equation}
\end{lemma}

\begin{proof}
	($\impliedby$) We show the existence of a system~\eqref{eq:UIO_dynamics_general} with matrices\vspace{-0.1cm}
	\begin{equation}\label{eq:UIO_matrices}\vspace{-0.1cm}
	\begin{split}
	A_\mathrm{UIO}=X_f\Xi_{X_p},~B_\mathrm{UIO}&=X_f(\Xi_{V_p} + \Xi_{X_p}X_f\Xi_{V_f}),\\
	D_\mathrm{UIO}&=X_f\Xi_{V_f}.
	\end{split}
	\end{equation}
	Note that every compatible trajectory is a sequence of input-output data $v_t$ and $x_t$,
	and verifies~\eqref{eq:compatible_trajectory}. If~\eqref{eq:Hankel_existence_condition} holds, the vector $x_{t+1}$ is uniquely determined by $X_fg_{t+1}$ for any vector $g_{t+1}$ fulfilling~\eqref{eq:g_sol_aff}.
	Therefore, for any compatible trajectory and $t$, $x_{t+1}$ is given by \vspace{-0.1cm}
	\begin{equation}\label{eq:xtp1_linear_func}\vspace{-0.1cm}
		x_{t+1} = X_f\Xi_{V_p}v_t + X_f\Xi_{X_p}x_t + X_f\Xi_{V_f}v_{t+1},
	\end{equation}
	since $\nu\in\ker([V_p^\top~X_p^\top~V_f^\top]^\top)\subseteq\ker(X_f)$. On defining $z_{t+1} \triangleq X_f\Xi_{V_p}v_t + X_f\Xi_{X_p}x_t$ and replacing the time index $t+1$ with $t$, equation~\eqref{eq:xtp1_linear_func} reduces to $x_t=z_t+X_f\Xi_{V_f}v_t$ which is the output equation in~\eqref{eq:UIO_dynamics_general} with $D_\mathrm{UIO}$ in~\eqref{eq:UIO_matrices}. Replacing $x_t$ with $z_t+X_f\Xi_{V_f}v_t$ in the definition of $z_{t+1}$ yields the state update in~\eqref{eq:UIO_dynamics_general} with $A_\mathrm{UIO}$ and $B_\mathrm{UIO}$ matrices in~\eqref{eq:UIO_matrices}. As such, the relation~\eqref{eq:xtp1_linear_func} between the elements of the tuple $(v_t, x_t, v_{t+1}, x_{t+1})$ is equivalently represented as the relation between the inputs and outputs of the system in~\eqref{eq:UIO_dynamics_general} with the matrices $(A_\mathrm{UIO},B_\mathrm{UIO},D_\mathrm{UIO})$ in~\eqref{eq:UIO_matrices} and the initial state $z_0 = x_0-D_\mathrm{UIO}v_0$. 
	
	($\implies$) Note that the system in~\eqref{eq:UIO_dynamics_general} generates all trajectories compatible with the historical data; therefore, the columns of $[V_p^\top~X_p^\top~V_f^\top~X_f^\top]^\top$ represent input-output trajectories of this system. Denote its corresponding historical state data by $\bar{z}\triangleq\mathrm{col}(\{\bar{z}_i\}_{i=0}^{T-1})$, which define the matrices $Z$, $Z_p$, and $Z_f$ as in~\eqref{eq:Hankel_IOS},~\eqref{eq:past_future_Hankel}. Since it holds that $Z_f = A_\mathrm{UIO}Z_p + B_\mathrm{UIO}V_p$, $X_p = Z_p + D_\mathrm{UIO}V_p$, and $X_f = Z_f + D_\mathrm{UIO}V_f$, one gets\vspace{-0.1cm}
	\begin{equation*}\vspace{-0.1cm}
		\begin{split}
		X_f =& (B_\mathrm{UIO}-A_\mathrm{UIO}D_\mathrm{UIO})V_p + A_\mathrm{UIO}X_p + D_\mathrm{UIO}V_f \\
		=& \begin{bmatrix}
		B_\mathrm{UIO}-A_\mathrm{UIO}D_\mathrm{UIO} & A_\mathrm{UIO} & D_\mathrm{UIO}
		\end{bmatrix}\begin{bmatrix}
		V_p \\
		X_p \\
		V_f
		\end{bmatrix}.
		\end{split}
	\end{equation*}
	This, in turn, implies~\eqref{eq:Hankel_existence_condition}, concluding the proof.
\end{proof}

Next, we discuss the existence of a UIO and provide a data-driven unknown-input state-estimation scheme. 

\begin{theorem}[Data-driven UIO]\label{thm:UIO_and_estimations} 
	Suppose that Assumption~\ref{ass:pers_excit} holds. There exists a UIO of the form~\eqref{eq:UIO_dynamics_general} with the matrices in~\eqref{eq:UIO_matrices} 
	if and only if~\eqref{eq:Hankel_existence_condition} holds and $X_f\Xi_{X_p}$ is Schur stable. Moreover, for any $\hat{x}_0\in\RR{n}$, the state estimations $\hat{x}_{t+1}$, $t=0,1,\dots$ computed through the iterative formula\vspace{-0.1cm}
	\begin{equation}\label{eq:iterative_xhat}\vspace{-0.1cm}
		\hat{x}_{t+1} = X_f\Xi[u_t^\top~
		y_t^\top~
		\hat{x}_t^\top~
		u_{t+1}^\top~
		y_{t+1}^\top]^\top
	\end{equation}
	asymptotically converge to the state $x_{t+1}$ of $\Gcal$.
\end{theorem}

\begin{proof}
	($\impliedby$) When Assumption~\ref{ass:pers_excit} and condition~\eqref{eq:Hankel_existence_condition} are satisfied, Lemmas~\ref{lem:historical_data_span} and~\ref{lem:ExistsSchurSystem} guarantee that the system~\eqref{eq:UIO_dynamics_general} with matrices given in~\eqref{eq:UIO_matrices} can generate any compatible trajectory, hence, any trajectory of $\Gcal$. Next, we focus on the iterative process~\eqref{eq:iterative_xhat} of computing estimations $\hat{x}_t$ from an initial condition $\hat{x}_0$ for any input $u$ and unknown input $d$. As described in the proof of Lemma~\ref{lem:ExistsSchurSystem}, this process is equivalent to generating output trajectories of the system in~\eqref{eq:UIO_dynamics_general} with the initial state $z_0 = \hat{x}_0 - D_\mathrm{UIO}[u_0^\top~y_0^\top]^\top$ and inputs $v_t = [u_t^\top~y_t^\top]^\top$. That proof also shows that the actual state $x_t$ of $\Gcal$ corresponds to the output of the same system with the same inputs but a different initial state: $z_0^\prime = x_0 - D_\mathrm{UIO}[u_0^\top~y_0^\top]^\top$. The state estimation error $e= x-\hat{x}$ is the difference between these two outputs of~\eqref{eq:UIO_dynamics_general}, which follows the autonomous dynamics $e_{t+1} = A_\mathrm{UIO}e_{t}$. If $A_\mathrm{UIO}$ is Schur stable, this error converges to zero and the LTI system in~\eqref{eq:UIO_dynamics_general} is a UIO by Definition~\ref{def:UIO}.
	
	($\implies$) From Definition~\ref{def:UIO} and Remark~\ref{rem:UIO_traj}, a UIO has Schur stable dynamics. Using Lemma~\ref{lem:ExistsSchurSystem}, existence of a UIO of the form~\eqref{eq:UIO_dynamics_general} implies~\eqref{eq:Hankel_existence_condition}.
\end{proof}

	Note that all $\Xi$ matrices such that $g_{t+1}$ in~\eqref{eq:g_sol_aff} verifies~\eqref{eq:g_sol} can be characterized as $\Xi = \Xi_0 + \Delta$, where $\Xi_0=([V_p^\top~X_p^\top~V_f^\top]^\top)^\dagger$ and $\Delta$ is any matrix such that $\mathrm{range}(\Delta)\subseteq\ker([V_p^\top~X_p^\top~V_f^\top]^\top)$. Under~\eqref{eq:Hankel_existence_condition}, it also holds that $\mathrm{range}(\Delta)\subseteq\ker(X_f)$. This implies that whether a UIO exists and, if yes, its matrices in~\eqref{eq:UIO_matrices}, are independent of the particular choice of $\Xi$.

\begin{remark}\label{rem:dd_vs_model-based}
	Unlike the proposed data-driven UIO, existing model-based design procedures provide a degree of freedom in choosing UIO matrices~\cite{valcher1999state, chen1996design}, which can be exploited to tune the estimation performance. Therefore, our UIO with matrices~\eqref{eq:UIO_matrices} corresponds to one specific choice that can be achieved using model-based design methods.
\end{remark}

\begin{remark}\label{rem:noisy_recent_data}
	If recent data $\{v_i\}_{i=0}^{N-1}=\{[u_i^\top~y_i^\top]^\top\}_{i=0}^{N-1}$ are affected by noise, the recursive algorithm~\eqref{eq:iterative_xhat} results in estimation errors. Note that~\eqref{eq:iterative_xhat} is equivalent to computing output trajectories of the UIO~\eqref{eq:UIO_dynamics_general} with matrices in~\eqref{eq:UIO_matrices}. Therefore, the noise in recent data acts as an input disturbance to~\eqref{eq:UIO_dynamics_general}, i.e., $\tilde{v}_t = v_t + w_t$ is applied as input instead of $v_t$, where $w_t$ is the measurement noise. Standard LTI system theory can be used to analyze the estimation error, which is the perturbation on the output of~\eqref{eq:UIO_dynamics_general} caused by $w_t$.
\end{remark}

We next provide an application example to demonstrate the use of the proposed method on DCmGs. We also show that it can be used for distributed cyber-attack detection.\vspace{-0.2cm}

\section{Distributed State Estimation in DCmGs} \label{sec:MicrogridExample}

An mG is an electrical network of distributed generation units (DGUs) and loads, capable to work either in grid-connected or islanded mode. Islanded mGs are usually controlled via hierarchical control schemes, where the primary controllers, often decentralized~\cite{nahata2020passivity}, provide voltage regulation, and higher-level controllers perform DGU coordination through a distributed architecture utilizing a communication network~\cite{tucci2018stable}. Network links can be compromised by cyber attacks. A distributed cyber-attack detection scheme comprising attack monitors collocated with every DGU has been proposed in~\cite{gallo2018distributed}. The key ingredients of local monitors are UIOs, used for estimating the state of neighboring DGUs. Hereafter, we use the proposed data-driven UIOs to replace the model-based ones in~\cite{gallo2018distributed} and show their effectiveness. This would eliminate the need for constructing accurate models of DGUs, which can be costly or require expertise.

\begin{figure}[t]
	\centering
	\ctikzset{bipoles/length=0.65cm}
	\tikzstyle{every node}=[font=\tiny]
	\begin{tikzpicture}[american currents, scale=0.5]
	\draw (1.5,4)
	to [short](1.5,4.5)
	to [short](3.5,4.5)
	to [short](3.5,0.5)
	to [short](1.5,0.5)
	to [short](1.5,4)
	to [short](1,4)
	to [battery, o-o](1,1)
	to [short](1.5,1)
	to [short](1,1);
	\node at (2.5,2.5){\footnotesize \textbf{Buck $i$}};
	\draw[-latex] (4,1.25) -- (4,3.75)node[midway,right]{$V_{ti}$};
	\draw (3.5,4) to [short](4,4)
	to [short](4.5,4)
	to [R=$R_{ti}$] (6,4)
	to [L=$L_{ti}$] (7.5,4)
	to [short, i=$I _{ti}$, -] (8.5,4)
	to [short](9,4)
	to [C, a=$C_{ti}$, -] (9,1)
	to [short](4,1)
	to [short](3.5,1);
	\draw (12.3,4)  to [R=$R_{ij}$] (14.7,4) 
	to [short, i=$I _{ij}$, -] (14.8,4)
	to [L=$L_{ij}$] (16.5,4)
	to [short, -o] (17.1,4) node[anchor=north,above]{$PCC_j$};
	\draw (8.5,4) to [short,-o] (11,4)
	to [dcisource, I ,a=$I_{Li}$] (11 ,1)
	to [short] (9,1);
	\draw[-latex] (11.5,1.25) -- (11.5,3.75)node[midway,right]{$V_i$};
	\draw (11,1)
	to [short, -o] (17.1,1);
	\draw (11,4) to [short](11.5,4);
	\draw (11,4) node[anchor=north, above]{$PCC_i$}  to [short](11,2.9);
	\draw (11,4) to [short](12.5,4);
	\draw[black, dashed] (.5,.25) -- (12.45,.25) -- (12.45,5.5) -- (.5,5.5)node[sloped, midway, above]{{\footnotesize \textbf{DGU and Load $i$ }}}  -- (.5,.25);
	\draw[black, dashed] (12.7,.25) -- (16.3,.25) -- (16.3,5.5) -- (12.7,5.5)node[sloped, midway, above]{{\footnotesize \textbf{Power line $ij$}}}  -- (12.7,.25);
	\end{tikzpicture}
	\caption{Electrical scheme of $i^{th}$ DGU with connecting line(s).}
\label{fig:DGUi}\vspace{-0.5cm}
\end{figure}
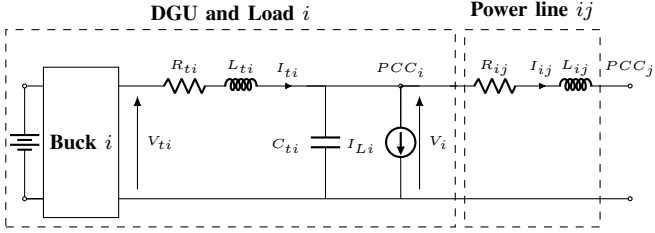

The electrical scheme of a DGU in a DCmG is given in Figure~\ref{fig:DGUi}, which defines relevant electrical parameters and variables (we refer the reader to~\cite{nahata2020passivity} for a comprehensive description of these quantities). When equipped with the primary controllers proposed in~\cite{nahata2020passivity}, the continuous-time dynamics of a DGU is $\dot{x} = A_cx + E_cd$, with $x \triangleq [V_i~I_{ti}~v_i]^\top$, $d = [I_{\mathrm{net},i}+I_{Li}~~V_{\mathrm{ref},i} + \alpha_i]^\top$, and \vspace{-0.15cm}
\small
\begin{equation}\label{eq:mg_matrices}\vspace{-0.05cm}
\begin{split}
&~~~A_c = \begin{bmatrix}
0 & \frac{1}{C_{ti}} & 0 \\
\frac{k_{i,1}-1}{L_{ti}} & \frac{k_{i,2}-R_{ti}}{L_{ti}} & \frac{k_{i,3}}{L_{ti}} \\
-1 & 0 & 0
\end{bmatrix},~
E_c = \begin{bmatrix}
-\frac{1}{C_{ti}} & 0 \\
0 & 0 \\
0 & 1
\end{bmatrix}.
\end{split}
\end{equation}
\normalsize
In particular, $k_{i,1}$, $k_{i,2}$, $k_{i,3}\in\RR{}$ are the parameters of the primary controller and $v_i$ is an integrator state introduced for penalizing the deviations of the output voltage $V_i$ from the reference $V_{\mathrm{ref},i}$~\cite{nahata2020passivity}. $I_{\mathrm{net},i}=\sum_{j\in\Ncal_i}I_{ij}$ is the net current injected into the mG by DGU $i$, where $\Ncal_i$ is the set of neighbors of DGU $i$\footnote{Neighbors are DGUs connected to DGU $i$ via a power line (see Figure~\ref{fig:DGUi}).}. Moreover, $\alpha_i$ is the output of a distributed secondary controller~\cite{tucci2018stable}. As in~\cite{gallo2018distributed}, we assume all states are measured and transmitted to the neighboring units. The unknown inputs can be measured; however, they are not sent to the neighboring units for security and privacy reasons. Indeed, transmitting these variables in real time would make them vulnerable to cyber-attacks, thus compromising the purpose of attack detection. Moreover, sharing historical data $\bar{d}_i$ with neighboring units might cause privacy violations. Indeed, the loads $I_{Li}$ often correspond to consumption, which can reveal the occupancy and daily activities of the consumers~\cite{hart1992nonintrusive}. Furthermore, the variables $I_{\mathrm{net},i}$ and $\alpha_i$ may contain sensitive information regarding the neighbors of DGU $i$, which might not be desirable to share. 

By using exact discretization, the discrete-time model of a DGU is given by system~\eqref{eq:LTI_dynamics} with\footnote{Hereafter, we omit the subscript $i$ as it is irrelevant for the UIO design.}\vspace{-0.2cm}
\begin{equation}\label{eq:mg_matrices_discrete}
	A = e^{A_cT_s}, ~ B = \0, ~ E = \left(\int_{\tau=0}^{T_s}e^{A\tau}d\tau\right)E_c, ~ C=\I 
\end{equation}
for a sampling period $T_s>0$, which we assume to be $10~ms$ in our experiments. 
At each time step $t$, the neighboring DGU $j$ receives the following communicated output from DGU $i$\vspace{-0.3cm}
\begin{equation}\label{eq:y^c}
y^c_t = y_t + \phi_t,
\end{equation}
where $\phi_t$ is the additive cyber attack vector at time $t$. $T_a$ denotes the start of the attack; therefore, $\phi_t$ is zero for all $t<T_a$, and non-zero for, at least, a time instant $t\geq T_a$. 

As in~\cite{gallo2018distributed}, we are interested in building a monitor collocated with the neighbor $j$ of DGU $i$, that estimates the states $x$ of DGU $i$ from the communicated outputs $y^c$ by assuming \textit{safe operation}. i.e., that there are no attacks, and therefore $y_t^c=y_t$. This corresponds to the problem of designing a UIO for the system in~\eqref{eq:LTI_dynamics} with the matrices in~\eqref{eq:mg_matrices_discrete}. 

We collect historical data by initializing the DGU from a random state. These data are not affected by attacks, as they are collected and sent to the neighboring units offline (see Remark~\ref{rem:available_data}). As discussed in Remark~\ref{rem:d_pers_excit}, it is sufficient that $V_{\mathrm{ref},i}$ and $I_{Li}$ have stochastic components to verify Assumption~\ref{ass:pers_excit}. This can indeed be satisfied as $V_{\mathrm{ref},i}$ is a free variable and $I_{Li}$ is the load current, which can be assumed to have a stochastic element as discussed in Remark~\ref{rem:d_pers_excit}.

The historical data verifies the conditions in Theorem~\ref{thm:UIO_and_estimations} for the existence of a UIO; therefore,
\eqref{eq:iterative_xhat} can be used to compute state estimates. 
This is expected, since a model-based UIO also exists for the same system~\cite{gallo2018distributed}.
We initialize the DCmG from a random initial condition, and simulate it for $N=10$ time steps with no attack and $d=d_0+\delta d$, where $d_0$ is a nominal vector and $\delta d$ is a small random component.
As shown in Figure~\ref{fig:safe_states}, the estimates quickly converge to the real states. 
In view of Remark~\ref{rem:dd_vs_model-based}, the same UIO estimations can also be obtained by a model-based design procedure in case DGU matrices~\eqref{eq:mg_matrices} are known.

\begin{figure}[t]
	\centering
	\includegraphics[width=0.45\textwidth]{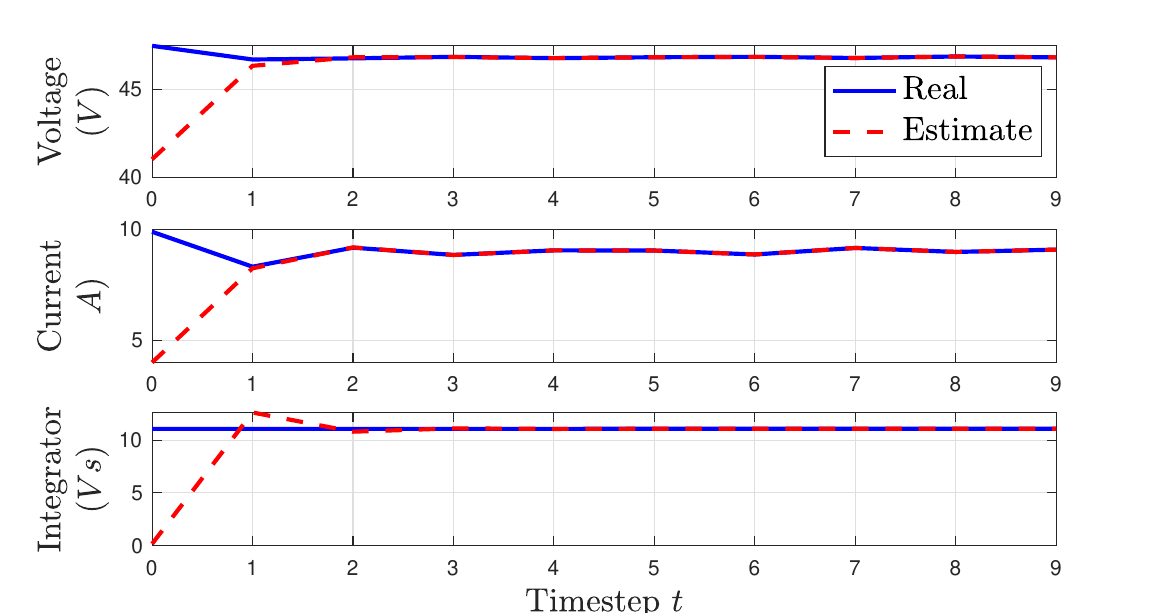}
	\caption{States and estimates in safe operation.} \label{fig:safe_states}\vspace{-0.45cm}
\end{figure}

\begin{figure}[t]
	\centering
	\includegraphics[width=0.45\textwidth]{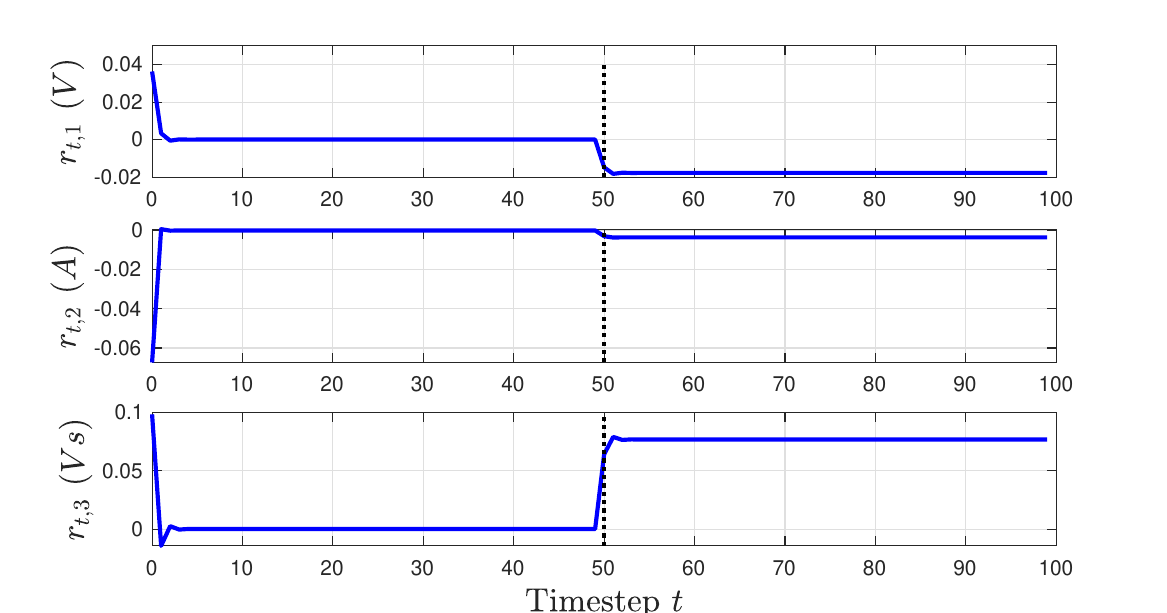}
	\caption{Residual signals in presence of attack. The vertical dotted line represents the start of the attack.} \label{fig:attacked_residuals}\vspace{-0.6cm}
\end{figure}

As shown in Lemma~1 in~\cite{gallo2020distributed}, it is possible that a UIO cannot detect any attack. We next introduce an attack in $y^c_t$ to illustrate that the data-driven UIO designed above can detect at least one attack and be used in the distributed cyber-attack detection scheme in~\cite{gallo2018distributed}. Using the same historical data, we run another simulation of length $N=100$ timesteps with random $x_0$ and $d$. Differently from the first case, a constant attack $\phi_t=[0.1~0.1~0.1]^\top$ is added on the communicated output variables in~\eqref{eq:y^c} after an attack start time of $T_a=50$. In this case, the output estimation error, called \textit{residual}, can be computed from the information available at the DGU $j$ as $r_t = y_t^c - \hat{x} = [r_{t,1}~r_{t,2}~r_{t,3}]^\top$. Figure~\ref{fig:attacked_residuals} demonstrates that the residuals are affected by the attack, showing the potential of the proposed method in distributed cyber-attack detection. \vspace{-0.20cm}

\section{Conclusions and Perspectives}\label{sec:Conclusions}
In this paper, we provide data-driven necessary and sufficient conditions for the existence of a UIO for an LTI system and propose a data-driven unknown-input state-estimation method. We also show the effectiveness of the algorithm for distributed state estimation in DCmGs. Future research directions include development of a completely data-driven attack-detection scheme using the proposed UIO, and the extension of the results to the case of noisy data.\vspace{-0.1cm}

\printbibliography
\end{document}